\documentclass[aps,prl,twocolumn,showpacs,superscriptaddress]{revtex4}

\usepackage{graphicx}
\usepackage{dcolumn}
\usepackage{amsmath}
\usepackage{enumerate,amsthm,amssymb,color}
\usepackage{bbm}
\usepackage{bm}
\usepackage[french, english]{babel}
\usepackage{synttree}

\newcommand{\tr}{\mathrm{ tr }}
\newcommand{\ie}{{\em i.e.} }

\newcommand{\ket}[1]{|#1\rangle}

\newcommand{\ketbra}[2]{|#1\rangle\langle#2|}
\newcommand{\fzz}{\mbox{$\ketbra{\Phi_0^{n-1}}{\Phi_0^{n-1}}$}}
\newcommand{\fzo}{\mbox{$\ketbra{\Phi_0^{n-1}}{\Phi_1^{n-1}}$}}
\newcommand{\foz}{\mbox{$\ketbra{\Phi_1^{n-1}}{\Phi_0^{n-1}}$}}
\newcommand{\foo}{\mbox{$\ketbra{\Phi_1^{n-1}}{\Phi_1^{n-1}}$}}
\newcommand{\fpp}{\mbox{$\ketbra{\Phi_+^{n-1}}{\Phi_+^{n-1}}$}}
\newcommand{\fmm}{\mbox{$\ketbra{\Phi_-^{n-1}}{\Phi_-^{n-1}}$}}

\newcommand{\altketbra}[1]{|#1\rangle \langle #1|}

\newcommand{\braket}[2]{\langle#1|   #2\rangle}

\newcommand{\COMMENT}[1]{}

\newtheorem{theorem}{Theorem}
\newtheorem{corollary}{Corollary}

\begin{document}

\preprint{APS/123-QED}

\title{Multipartite Entanglement Verification Resistant against Dishonest Parties}

\author{Anna Pappa}\affiliation{LTCI, CNRS - T\'el\'ecom ParisTech, Paris 75013, France}\affiliation{LIAFA, CNRS - Universit\'e Paris 7, Paris 75013, France}
\author{Andr\'e Chailloux}\affiliation{Computer Science Department, University of California, Berkeley, California 94720-1776, USA}
\author{Stephanie Wehner}\affiliation{Center for Quantum Technologies, National University of Singapore, Singapore 117543}
\author{Eleni Diamanti}\affiliation{LTCI, CNRS - T\'el\'ecom ParisTech, Paris 75013, France}
\author{Iordanis Kerenidis}\affiliation{LIAFA, CNRS - Universit\'e Paris 7, Paris 75013, France}\affiliation{Center for Quantum Technologies, National University of Singapore, Singapore 117543}

\date{\today}

\begin{abstract}
Future quantum information networks will consist of quantum and classical agents, who have the ability to communicate in a variety of ways with trusted and untrusted parties and securely delegate computational tasks to untrusted large-scale quantum computing servers. Multipartite quantum entanglement is a fundamental resource for such a network and hence it is imperative to study the possibility of verifying a multipartite entanglement source in a way that is efficient and provides strong guarantees even in the presence of multiple dishonest parties.
In this Letter, we show how an agent of a quantum network can perform a distributed verification of a source creating multipartite Greenberger-Horne-Zeilinger (GHZ) states with minimal resources, which is, nevertheless, resistant against any number of dishonest parties. Moreover, we provide a tight tradeoff between the level of security and the distance between the state produced by the source and the ideal GHZ state. Last, by adding the resource of a trusted common random source, we can further provide security guarantees for all honest parties in the quantum network simultaneously.
\end{abstract}

\maketitle

Entanglement plays a key role in the study and development of quantum information theory. It has been widely used in all aspects of quantum information and has been essential to show the advantages obtained compared to the classical setting. Initially defined for bipartite states, the notion of entanglement has been generalized to multipartite systems and despite the complexity this notion acquires in this case, many interesting properties of multipartite entangled states are known. If we consider, for example, the quantum correlations of the Greenberger-Horne-Zeilinger (GHZ) state~\cite{GHZ} and its $n$-party generalization, we can find a nonlocal game that can be won with probability 1 in the quantum setting, while any classical local theory can win the game with probability at most $3/4$~\cite{mermin:prl90}.

Multipartite entangled states are a fundamental resource when quantum networks are considered. Indeed, they allow network agents to create strong correlations in order  to perform distributed tasks, to delegate computation to untrusted servers~\cite{broadbent:focs09}, or to compute, for example through the Measurement-Based Quantum Computation model~\cite{raussendorf:prl01}. A natural and fundamental question that arises then is whether the network agents should be required to trust the source that provides them with such multipartite entangled states or whether they are able to verify the entanglement.

In this work, we show that a quantum agent can verify efficiently with respect to the necessary resources, that an untrusted source creates entanglement, even in the presence of dishonest parties.

{\bf The model} -- We start our analysis by first describing in detail our model and its relation to previous work.
\paragraph{Source:} The source is untrusted. It is supposed to create the $n$-party GHZ state $\frac{1}{\sqrt{2^n}}\big(|0^n\rangle+|1^n\rangle\big)$
and distribute it to $n$ parties. By applying a Hadamard and a phase shift ($\sqrt{Z}$) gate to each qubit, the GHZ state can be expressed by the locally equivalent state
\begin{equation*}
\lvert\Phi_0^n\rangle=\frac{1}{\sqrt{2^{n-1}}}\Big[\sum_{\Delta(y)\equiv0\pmod{4}}\lvert y \rangle - \sum_{\Delta(y)\equiv2\pmod{4}}\lvert y \rangle\Big],
\end{equation*}
where $y$ is a classical $n$-bit string $y_1...y_n$ and $\Delta(y) = \sum_i y_i$ denotes the Hamming weight of $y$. We will use the latter state for our proofs. Such states have a wide range of applications, for example in nonlocal games, quantum game theory and quantum computation.

\paragraph{Parties:} A party can be honest or dishonest. An honest party does not know which parties are honest and which are dishonest. The dishonest parties can collaborate with each other and control the source. Their goal is to convince the honest parties that the source can create the $n$-party GHZ state, while in reality this may not be true. They are allowed to create a different state every time or entangle the state with any auxiliary space.
\paragraph{Local resources:} A party has a trusted single-qubit measurement device with two measurement settings and a trusted classical random source.
\paragraph{Network resources:} Every pair of parties shares a private classical channel, in other words, the communication between two honest parties remains secret. This is the standard setup for classical networks with dishonest parties, since in the absence of private channels we cannot guarantee security for more than a single honest party.

Note that we want to use the least possible resources for our protocol. Indeed, we only need each party to be able to produce randomness, to perform single-qubit measurements and to securely communicate classical information with the other parties. Since our goal is to construct protocols that can be widely available in the near future, it is imperative to minimize the resources available to the agents, especially the quantum resources that are considered more expensive than the classical ones, hence bringing such tasks closer to reality.

{\bf Related work} -- Most of the work on entanglement verification has considered the case where all parties are honest. For two parties, three models have been studied.

First, the standard model, where both parties trust their devices but they do not trust the source. This model corresponds to the setting of non-separability tests, where the two parties can perform together quantum tomography on the state distributed by the source and thus verify the existence of entanglement. In a cryptographic language, this corresponds to a setting where all parties are guaranteed to be honest. A related question concerning untrusted sources in quantum key distribution protocols is discussed in \cite{Lo:pra10}.

Second, the device independent model, where the parties trust neither their quantum devices nor the source. This model is related to the well known setting of the Bell nonlocality tests as well as to self-testing \cite{McKague:arxiv12}.

Third, the one-sided device independent model~\cite{branciard:arxiv11}, where the security is asymmetric: one party trusts his devices but the second party's devices and the source are untrusted. This model corresponds to the setting of generalized quantum steering~\cite{wiseman:prl07,jones:pra07}, where one party is also given control of the source and tries to convince the other party, who trusts his devices, that she can create entanglement. In a cryptographic language, an honest party tries to verify entanglement in the presence of a dishonest party who controls the source. Recently, there have been experimental demonstrations in this model~\cite{bennet:arxiv11, wittmann:arxiv11, smith:arxiv11}.

For the multipartite case, much less is known. In the standard model, pseudo-telepathy~\cite{brassard:iwads03} extends Mermin's game~\cite{mermin:prl90} to many parties; a maximally entangled state is used to play the game and wins with probability one, which is strictly better than in the classical case. In the device independent model, it was shown that honest parties who do not trust their devices can verify genuine multipartite entanglement by using appropriate entanglement witnesses~\cite{bancal:prl11}. Finally, in~\cite{cavalcanti:pra11} the authors present a unified framework for $n$-party entanglement verification and provide inequalities with different bounds for the different nonlocality classes that are considered.

{\bf Our Work} --
In our model, there are, in general, $k$ honest parties and $n-k$ dishonest parties who control the source, which is supposed to create an $n$-party GHZ state. Each honest party does not know which other parties are honest. Our goal is to provide an efficient test for an honest party, such that the test passes only if the state produced by the source creates entanglement between all $k$ honest parties and the $n-k$ dishonest parties.

First, if all players are honest, we prove that  any $n$-party state that is $\epsilon$-away from the $n$-party GHZ state passes the test with probability at most $1-\epsilon^2/2$. Second, in the presence of any number of dishonest parties, we prove a similar quantitative statement, this time for any $n$-party state that is $\epsilon$-away from the $n$-party GHZ state  up to a local unitary operation on the space of the dishonest parties.

For the special case of $n=2$, our model significantly extends the results in the generalized quantum steering setting by providing a tight analysis of the tradeoff between the distance of the shared state to the $n$-party GHZ and the probability of success of the test. For the case of $n=k$, \ie the standard model, our results again provide a tight analysis of the tradeoff between the distance and the probability the test passes. For general $n$ parties and $k$ honest ones, this is the first rigorous analysis of an entanglement verification test.

{\bf The Protocol $V$} -- Consider a source that is supposed to create and distribute the state $\ket{\Phi_0^n}$ to $n$ parties. We present a verification protocol $V$ that one party, called the Verifier, can run with the other $n-1$ parties, in order to verify that the state $\ket{\Psi}$ created by the source is in fact the correct one.
\begin{enumerate}
\item The Verifier selects for each $i\in [n]$ a random input $X_i\in\{0,1\}$, such that $\sum_{i=1}^n X_i \equiv 0 \pmod{2}$, and sends it to the corresponding party via a private classical channel.
\item  If  $X_i=0$, party $i$ performs a $Z$ operation.\\
         If  $X_i=1$, party $i$ performs a Hadamard operation.
\item Party $i$ measures in the $\{|0\rangle,|1\rangle\}$ basis and sends the corresponding outcome $Y_i\in\{0,1\}$ to the Verifier via the private channel.
\item The Verifier accepts the result if and only if:
\begin{equation*}
\sum_{i=1}^nY_i\equiv \frac{1}{2}\sum_{i=1}^nX_i \pmod{2}
\end{equation*}
\end{enumerate}

The above protocol assumes that a specific party plays the role of the Verifier. We will later address the question of how to pick such a Verifier among all honest parties.
Note that this test has been used before \cite{brassard:iwads03}, however our analysis is entirely different.
We denote by $T(\ket{\Psi})$ the event that the Verifier accepts the result of the Test, when the joint state is $\ket{\Psi}$.

{\bf Correctness of the protocol} --  We want to show that the state $\ket{\Phi_0^n}$ passes the Test with probability 1. We need to define the following state:
\begin{equation*}
|\Phi_1^n\rangle=\frac{1}{\sqrt{2^{n-1}}}\Big[\sum_{\Delta(y)\equiv1\pmod{4}}\lvert y \rangle - \sum_{\Delta(y)\equiv3\pmod{4}}\lvert y \rangle\Big]
\end{equation*}
It is easily verifiable (from the definition of the states) that for any $k$ and $n$:
\begin{equation}\label{eq:nparty}
\lvert\Phi_0^n\rangle=\frac{1}{\sqrt{2}}\Big[\lvert\Phi_0^k\rangle\lvert\Phi_0^{n-k}\rangle-\lvert\Phi_1^k\rangle\lvert\Phi_1^{n-k}\rangle\Big]
\end{equation}
From condition $\sum_{i=1}^n X_i\equiv0\pmod{2}$, we have two cases:
\begin{itemize}
\item $(\frac{1}{2}\sum^n_{i=1}X_i)\equiv 0\pmod 2$: This means that the sum of the inputs is a multiple of 4. Using Eq. (\ref{eq:nparty}), it can be proven that the state $|\Phi_0^n\rangle$ goes to $\pm|\Phi_0^n\rangle$ when we apply to it an operator consisting of a $0\pmod4$ number of single-qubit Hadamards, and $Z$ gates on the remaining qubits. Hence we always have $\sum_{i=1}^nY_i\equiv0\pmod{2}$.
\item $(\frac{1}{2}\sum^n_{i=1}X_i)\equiv 1$(mod 2): This means that the sum of the inputs is even but not a multiple of 4. Again, it can be proven that  the state $|\Phi_0^n\rangle$ goes to  $\pm|\Phi_1^n\rangle$ when we apply to it an operator consisting of a $2\pmod4$ number of single-qubit Hadamards, and $Z$ gates on the remaining qubits. Hence we always have $\sum_{i=1}^nY_i\equiv1\pmod{2}$.
\end{itemize}

{\bf Security in the Honest Model} -- We now look at the model where all $n$ parties are honest and analyze the probability that our Test accepts a state as a function of the distance of this state to $\ket{\Phi_0^n}$. We first analyze the case of a pure state. Denoting by $D(|\psi\rangle,|\phi\rangle)$ the trace distance between two states $|\psi\rangle$ and $|\phi\rangle$, we have
\begin{theorem}
If $D(|\Psi\rangle,|\Phi_0^n\rangle)=\epsilon$,  $\Pr [ T(\ket{\Psi}) ] \leq 1-\frac{\epsilon^2}{2}$.
\end{theorem}
The main idea of the proof is to show that our Test is equivalent to performing a POVM $\{P_n,I-P_n\}$ (where the first outcome corresponds to acceptance) with
\begin{equation*}
 P_n = \ketbra{\Phi^n_0}{\Phi^n_0} + \frac{1}{2}I_{S_n},
\end{equation*}
where $S_n$ denotes the subspace of  $n$-qubit states that are orthogonal to both $|\Phi_0^n\rangle$ and $|\Phi_1^n\rangle$, and $I_{S_n}$ denotes the projection on this subspace. In other words, we show that the state $|\Phi_0^n\rangle$ passes the Test with probability 1, the state $|\Phi_1^n\rangle$ passes the Test with probability 0 and all other states in the orthogonal subspace pass the Test with probability exactly $1/2$. The proof of this statement is in fact quite involved and is done by induction on the dimension of the state (see Appendix for details). With this characterization for our Test, we express any state $\ket{\Psi}$ such that $D(\ket{\Psi},\ket{\Phi_0^n})=\epsilon$ as
$\lvert\Psi\rangle=\sqrt{1-\epsilon^2}|\Phi_0^n\rangle+\sum_{i=1}^{2^n-1}\epsilon_i|\Phi_i^n\rangle,
$
where for $i\geq2,~|\Phi_i^n\rangle\in S_n$ and $\sum_{i=1}^{2^n-1}\epsilon_i^2=\epsilon^2$, and hence we have
$\Pr[ T(\ket{\Psi})] =   \tr(P_n |\Psi\rangle) \leq 1-\frac{\epsilon^2}{2}.
$

Note that for a mixed state $\rho=\{p_i,\ket{\Psi_i}\}$, $\Pr[T(\rho)]= \sum_i p_i \Pr[T(\ket{\Psi_i})]$. Then, by convexity we have
\begin{corollary}
If $D(\rho,|\Phi_0^n\rangle)=\epsilon$,  $\Pr [ T(\rho) ] \leq 1-\frac{\epsilon^2}{2}$.
\end{corollary}

{\bf Security in the Dishonest Model} -- We look now at the model where the honest Verifier runs the test in the presence of dishonest parties. The Verifier is considered to be known to all parties. We prove in this case a theorem similar to the case of all honest parties. It should be clear here that there is no way for the honest parties to determine whether the dishonest parties act as $n-k$ independent parties each holding one qubit or whether they have colluded to one party. For example, the state $\lvert\Phi_0^{k+1}\rangle=\frac{1}{\sqrt{2}}\Big[\lvert\Phi_0^k\rangle\lvert 0\rangle-\lvert\Phi_1^k\rangle\lvert 1\rangle\Big]$, where the $n-k$ dishonest parties hold a single qubit, passes the Test with probability 1, since the dishonest parties can locally map this state to $\ket{\Phi_0^n}$. Hence, the correct security statement must take into account the fact that the dishonest parties may apply some operator on their space.
\begin{theorem}
Let $|\Psi\rangle$ be the state of all $n$ parties. If $\epsilon=\min_U D(U|\Psi\rangle,|\Phi_0^n\rangle)$, where $U$ is an operator on the space of the dishonest parties, then $\Pr[T(\ket{\Psi}] \leq1-\frac{\epsilon^2}{4}$.
\end{theorem}

Let us assume, without loss of generality, that the $n$ parties share a state of the form
$\lvert\Psi\rangle=|\Phi_0^k\rangle|\Psi_0\rangle+|\Phi_1^k\rangle|\Psi_1\rangle+\ket{\mathcal{X}}$,
where in $\ket{\mathcal{X}}$ the component of the honest parties is orthogonal to $\ket{\Phi_0^k}$ and $\ket{\Phi_1^k}$, and the dishonest state is not normalized. For the dishonest parties, making the Verifier accept the Test is equivalent to guessing the honest output $Y_H:=\sum_H Y_i\pmod 2$, where $H$ is the set of the honest parties, before announcing their measurement outcomes. The optimal probability of guessing $Y_H$ given $X_H:=\sum_H X_i\pmod 2$ is given by the Helstrom measurement. Let $p = \|\ket{\Psi_0}\|^2$, $q = \|\ket{\Psi_1}\|^2$ and $\braket{\Psi_0}{\Psi_1}^2= pq\cos^2\theta$, where $\theta$ is the angle between $\ket{\Psi_0}$ and $\ket{\Psi_1}$. Then, we calculate that:
\begin{eqnarray*}\label{eq3main}
Pr[\text{guess } Y_H]
  &\leq & 1 - \frac{1}{4}\Big(1-\frac{(p+q)^2+4pq\sin^2\theta}{2}\Big).
\end{eqnarray*}

We consider now that the dishonest parties can perform a local operation on their state, in order to maximize their cheating probability. Thus, the distance of the dishonest state from the correct one is : $\epsilon = min_UD((I\otimes U)\ket{\Psi},\ket{\Phi_0^n})=min_U\sqrt{1-F^2((I\otimes U)\ket{\Psi},\ket{\Phi_0^n})}$, where by $F(|\psi\rangle,|\phi\rangle)$ we denote the fidelity between two states $|\psi\rangle$ and $|\phi\rangle$. If the reduced density matrices of the honest players of the perfect and the real state are $\sigma_H$ and $\rho_H$ respectively, it holds that there exists a local operation $R$ on the dishonest state such that:
\[
F((I\otimes R)\ket{\Psi},\ket{\Phi_0^n}) = F(\sigma_H,\rho_H)=\frac{p+q+2\sqrt{pq}sin\theta}{2}
\]
For this operation $R$, the distance becomes:
\[
\epsilon^2 \le 1-F^2((I\otimes R)\ket{\Psi},\ket{\Phi_0^n}) = 1-F^2(\sigma_H,\rho_H),
\]
which concludes the proof (see Appendix for more details).

{\bf Security for all honest parties} --
We have presented a protocol that a Verifier can use to verify the state of an untrusted source in the presence of dishonest parties with minimal resources.
Our protocol can be useful in the scenario where some party wants to perform a complex quantum computation and needs to delegate parts of the computation to other parties, who, of course, would need some source of multipartite entanglement in order to perform the joint computation.  Note that the Verifier can repeat the protocol sequentially in order to increase the probability of detecting an erroneous state.

In a more general scenario, however, where parties need to perform securely some distributed multipartite computation using the multipartite entangled state as an initial shared resource, we need to guarantee security for all honest parties at the same time. In other words, we would like a protocol that guarantees to all honest parties that they will only accept to use a state for the computation that comes from a source that produces states that are very close to an $n$-party GHZ state.
A priori, such a task is impossible, since any such protocol could be used to produce unbiased strong coins \cite{lochau98} (the parties could just measure the entangled state to produce coins). Hence, we need to provide some additional resource.

\paragraph{Trusted Common Random Source (CRS):} We assume that all parties have access to a trusted classical random source that provides them with the same randomness.

This is, of course, a powerful, but necessary, resource. One way to achieve it would be to assume that at least a third of all parties are honest, since this implies the ability to securely produce random bits only with authenticated classical communication~\cite{chaum:stoc88}. Note that in order to achieve quantum secure multiparty computation, at least  a majority of honest parties is required \cite{ben-or2006}, in which case it is possible to construct a CRS.

We describe how to repeat our verification test in order to guarantee the following: when the parties decide to use the state for computation, then the probability that the state produced by the source is $\epsilon$-away from the $n$-party GHZ state goes to zero exponentially fast with the number of repetitions. Note that our guarantee is on the state produced by the source. Of course, we cannot prevent the dishonest parties from destroying the entanglement with the honest parties just before using this state for further computation. However, we argue that our test is still useful for secure multiparty computation. First, as we noted before, if the goal of the dishonest parties is to convince the honest parties of the source's ability to create entanglement, destroying the entanglement after the source has produced it does not help them. Second, in general, in secure multiparty computation, one of the main goals is to guarantee that the inputs of the honest parties remain secret for the dishonest parties.
Since in our model the parties will only perform local quantum operations (if the parties could send their qubits to other parties then checking the source will be much easier by having all qubits sent to the verifier), by destroying the entanglement, the dishonest parties cannot increase their information about the honest parties' inputs.  Third, we still have a strong guarantee on the honest players' states from which they can, for example, extract correlated secret bits.

Let $S$ be a security parameter.\\
\noindent
{\bf The Symmetric Protocol}
\begin{enumerate}
\item The source distributes a state $\ket{\Psi}$ to the $n$ parties (the honest source distributes the state $\lvert\Phi_0^n\rangle$).
\item Parties receive $r\in\{0,1\}^S$ and $i \in [n]$ from CRS.
\begin{enumerate}
\item If $r=\bf{0}$, the state $\ket{\Psi}$ is used for computation.
\item If $r \neq \bf{0},$ Party $i$ runs Protocol $V$ with $\ket{\Psi}$. If he rejects, then abort, otherwise go to Step 1.
\end{enumerate}
\end{enumerate}
Note that the source may create a different state at every repetition of the protocol. It is also important that the state is distributed before the parties receive the randomness from the CRS. Let $C_\epsilon$ be the event that the symmetric protocol has not aborted and that the state used for the computation, which we denote by $\ket{\Psi}$, is such that $\min_U D(U\ket{\Psi},\ket{\Phi_0^n}) \geq \epsilon$, where $U$ is an operator on the space of the dishonest parties.
We will prove the following:
\begin{theorem}
For all $\epsilon >0$,
$\Pr[C_\epsilon] \leq  2^{-S}\frac{4n}{k\epsilon^2}.$
\end{theorem}
The proof is given in the Appendix. When the Verifier is dishonest we suppose that the state always passes the test. By choosing $S= \log \frac{4n\delta}{\epsilon^2}$ for some constant $\delta>0$, all honest parties have the guarantee that the probability the state used has distance at least $\epsilon$ from the correct one, is at most $1/\delta$.
Note that the expected number of repetitions of the protocol is $2^S$, which, with our choice of $S$, is polynomial in $n$ and $1/\epsilon$ (and with probability exponentially close to 1 the number of repetitions is at most $O(2^S)$).
Moreover, this protocol provides guarantees to all honest parties, unlike the case of quantum steering and our multipartite generalization. To this end, it was necessary to make the assumption of a trusted classical random source.

{\bf Discussion} -- It is important to note that our analysis does not take into account losses and noise that appear in a realistic setting. It will be interesting to study such conditions, as was recently done for bipartite quantum steering~\cite{bennet:arxiv11}.
We also note that although our results provide a verification test for the GHZ state, the analysis should in principle be applicable to all states for which a Bell-type test is available, such as stabilizer states.

We acknowledge discussions with A. Marin, D. Markham, T. Lawson, and A. Leverrier, and financial support from the ANR, Digiteo, the EU and the Ministry of Education, Singapore.

\begin{appendix}

\begin{center}{\bf APPENDIX}\end{center}
{\bf Security in the Honest Model} \\

We look first at the model where all $n$ parties are honest and analyze the probability that our Test accepts a state as a function of the distance of this state to $\ket{\Phi_0^n}$. We first analyze the case of a pure state. Denoting by $D(|\psi\rangle,|\phi\rangle)$ the trace distance between two states $|\psi\rangle$ and $|\phi\rangle$, we have
\setcounter{theorem}{0}
\begin{theorem}
If $D(|\Psi\rangle,|\Phi_0^n\rangle)=\epsilon$,  $\Pr [ T(\ket{\Psi}) ] \leq 1-\frac{\epsilon^2}{2}$.
\end{theorem}
\begin{proof}

Let $X_1$ and $Y_1$ be the input and output of party 1, respectively. For $\sum_{i = 1}^n X_i \equiv 0\pmod 2$ and randomly chosen inputs, we have the following cases:

\begin{enumerate}
\item $X_1 = 0$ and $Y_1 = 0$, the test passes if\\
$\sum_{i=2}^{n} Y_i \equiv \left\{
\begin{array}{l}
0 \pmod 2 \text{ when } \sum_{i=2}^n X_i \equiv 0 \pmod 4 \\
1 \pmod 2 \text{ when }\sum_{i=2}^n X_i \equiv 2 \pmod 4 \\
\end{array}
\right.
$

\item $X_1 = 0$ and $Y_1 = 1$, the test passes if\\
$\sum_{i=2}^{n} Y_i \equiv \left\{
\begin{array}{l}
1 \pmod 2\text{ when } \sum_{i=2}^n X_i \equiv 0 \pmod 4 \\
0 \pmod 2\text{ when } \sum_{i=2}^n X_i \equiv 2 \pmod 4 \\
\end{array}
\right.$

\item $X_1 = 1$ and $Y_1 = 1$, the test passes if\\
$\sum_{i=2}^{n} Y_i \equiv \left\{
\begin{array}{l}
0 \pmod 2\text{ when } \sum_{i=2}^n X_i \equiv 1 \pmod 4 \\
1 \pmod 2\text{ when } \sum_{i=2}^n X_i \equiv 3 \pmod 4 \\
\end{array}
\right.$

\item $X_1 =1$ and $Y_1 = 0$, the test passes if \\
$\sum_{i=2}^{n} Y_i \equiv \left\{
\begin{array}{l}
1 \pmod 2\text{ when } \sum_{i=2}^n X_i \equiv 1 \pmod 4 \\
0 \pmod 2\text{ when } \sum_{i=2}^n X_i \equiv 3 \pmod 4 \\
\end{array}
\right.$

\end{enumerate}
We can consider the above tests on the $n-1$ parties as projective measurements for different inputs. When $\sum_{i = 2}^n X_i \equiv 0\pmod 2$, the $n-1$ parties perform a POVM $\{P_{n-1}, I - P_{n-1}\}$, where the first outcome corresponds to case 1 and the second to case 2. When $\sum_{i = 2}^n X_i \equiv 1\pmod 2$, the test is equivalent to performing a POVM $\{Q_{n-1}, I - Q_{n-1}\}$, where the first outcome corresponds to case 3 and the second to case 4.

Let $S_n$ the subspace of quantum pure states of $n$ qubits that are orthogonal to both $|\Phi_0^n\rangle$ and $|\Phi_1^n\rangle$, and let $I_{S_n}$ the projection on this subspace. We will prove by induction that:
\[
P_n = \ketbra{\Phi^n_0}{\Phi^n_0} + \frac{1}{2}I_{S_n}~,~Q_n = \ketbra{\Phi^n_+}{\Phi^n_+} + \frac{1}{2}I_{S_n}
\]
For $n = 1$, $S_n$ is the empty subspace, $P_1 = \ketbra{0}{0}$ and $Q_1 = \ketbra{+}{+}$, so the statements hold. We now suppose the statements are true for $n-1$ parties with $n \ge 2$ and we prove it for $n$ parties. Since the choice of $X_1$ is uniformly random, we have
\begin{eqnarray*}\label{P_n}
P_n & = &\frac{1}{2} \left( \ketbra{0}{0} \otimes P_{n-1}+ \ketbra{1}{1} \otimes (I - P_{n-1}) \right)\nonumber\\
       & + &\frac{1}{2} \left( \ketbra{+}{+} \otimes (I - Q_{n-1}) + \ketbra{-}{-} \otimes Q_{n-1} \right)
\end{eqnarray*}

We now use the induction property. Notice that $(I - P_{n-1}) = \ketbra{\Phi^{n-1}_1}{\Phi^{n-1}_1} + \frac{1}{2}I_{S_{n-1}}$ and $(I - Q_{n-1}) = \ketbra{\Phi^{n-1}_-}{\Phi^{n-1}_-} + \frac{1}{2}I_{S_{n-1}}$. Then,
\begin{eqnarray*}
P_n & = &\frac{1}{2}\left( \ketbra{0}{0} \otimes \fzz  + \ketbra{1}{1} \otimes \foo  \right. \\
       & + & \ketbra{+}{+} \otimes \fmm  + \ketbra{-}{-} \otimes \fpp \\
       & + & \left. I_1 \otimes I_{S_{n-1}} \right),
\end{eqnarray*}
where $I_1$ is the identity operator on one qubit.

Using $\fmm + \fpp = \fzz + \foo$ and $\fmm - \fpp = - \fzo - \foz$, we have:
\begin{eqnarray*}
P_n & = &\ketbra{0}{0}\Big(\frac{3}{4} \fzz + \frac{1}{4} \foo\Big) \\
      & - &\ketbra{0}{1}\Big(\frac{1}{4}\fzo + \frac{1}{4}\foz\Big) \\
      & - & \ketbra{1}{0}\Big(\frac{1}{4}\fzo + \frac{1}{4}\foz\Big) \\
      & + &\ketbra{1}{1}\Big(\frac{1}{4} \fzz + \frac{3}{4} \foo \Big) \\
      & + &\frac{1}{2} I_1 \otimes I_{S_{n-1}}
\end{eqnarray*}
Note that
\begin{align*}
\ket{\Phi_0^n} & = \frac{1}{\sqrt{2}} \ket{0}\ket{\Phi_0^{n-1}} - \frac{1}{\sqrt{2}}\ket{1}\ket{\Phi_1^{n-1}} \\
\ket{\Phi_1^n} & = \frac{1}{\sqrt{2}} \ket{0}\ket{\Phi_1^{n-1}} + \frac{1}{\sqrt{2}}\ket{1}\ket{\Phi_0^{n-1}}
\end{align*}
and define
\begin{align*}
\ket{\Psi_0^n} & = \frac{1}{\sqrt{2}} \ket{0}\ket{\Phi_0^{n-1}} + \frac{1}{\sqrt{2}}\ket{1}\ket{\Phi_1^{n-1}} \\
\ket{\Psi_1^n} & = \frac{1}{\sqrt{2}} \ket{0}\ket{\Phi_1^{n-1}} - \frac{1}{\sqrt{2}}\ket{1}\ket{\Phi_0^{n-1}}
\end{align*}
After some simple calculations, we have
\begin{align*}
P_n = \ketbra{\Phi^n_0}{\Phi^n_0} + \frac{1}{2} \ketbra{\Psi^n_0}{\Psi^n_0} + \frac{1}{2} \ketbra{\Psi^n_1}{\Psi^n_1} +  \frac{1}{2} I_1 \otimes I_{S_{n-1}}
\end{align*}
Note also that
\begin{align*}\label{ISn}
I_{S_n} = I_1 \otimes I_{S_{n-1}} + \ketbra{\Psi_0^n}{\Psi_0^n} +
\ketbra{\Psi_1^n}{\Psi_1^n},
\end{align*}
which concludes the proof for $P_n$. To complete the induction, we can show in a similar way that $ Q_n = \ketbra{\Phi^n_+}{\Phi^n_+} + \frac{1}{2}I_{S_n} $.

Assume now that the $n$ honest parties share a state $\ket{\Psi}$ with $D(\ket{\Psi},\ket{\Phi_0^n})=\epsilon$. We can express $\lvert\Psi\rangle$ as follows:
\begin{equation*}
\lvert\Psi\rangle=\sqrt{1-\epsilon^2}|\Phi_0^n\rangle+\sum_{i=1}^{2^n-1}\epsilon_i|\Phi_i^n\rangle,
\end{equation*}
where for $i\geq2,~|\Phi_i^n\rangle\in S_n$ and $\sum_{i=1}^{2^n-1}\epsilon_i^2=\epsilon$. The state $|\Psi\rangle$ verifies the Test in the protocol with probability $\tr(P_n |\Psi\rangle)$, so we have:
\begin{eqnarray*}
\Pr[ T(\ket{\Psi})]&=&1-\epsilon^2+\frac{\sum_{i=2}^{2^n-1}\epsilon^2}{2}\leq1-\epsilon^2 +\frac{\epsilon^2}{2}\\
    &=&1-\frac{\epsilon^2}{2}
\end{eqnarray*}
\end{proof}

\noindent
{\bf Security in the Dishonest Model} \\

We look now at the model where the honest Verifier runs the test in the presence of dishonest parties. The Verifier is considered to be known to all parties. Let us assume that there is a set $H$ of $k$ honest parties (the number and identity of all but himself are unknown to the Verifier) and a set $D$ of $n-k$ dishonest parties that collaborate with the quantum source. The goal of the dishonest parties is to convince the Verifier that the source can create the $n$-party GHZ state, while this may not be true.

Let $X_D :=  \sum_D X_i\pmod 2, Y_D := \sum_D Y_i\pmod 2$ denote the parity of the inputs and outputs of the dishonest parties, respectively (similarly for $X_H, Y_H$ for the honest parties). Since the Verifier picks the inputs for all parties at random, we can describe the two equally probable Tests that the $k$ honest parties perform with the POVMs $\{P_k, I-P_k\}$ and $\{Q_k,I-Q_k\}$, as they were defined before.

More precisely, if $X_D=0$ then the Test corresponds to the POVM $\{P_k, I-P_k\}$, where the accepting outcome is $P_k$ when  $Y_D\equiv\frac{1}{2}\sum_D X_i\pmod 2$ and $I-P_k$ when $ Y_D\not\equiv\frac{1}{2}\sum_DX_i\pmod 2$.
If $X_D=1$ then the Test corresponds to the POVM $\{Q_k,I-Q_k\}$, where the accepting outcome is $Q_k$ when  $Y_D\equiv\frac{1}{2}(-1+\sum_DX_i)\pmod 2$ and $I-Q_k$ when $Y_D\not\equiv\frac{1}{2}(-1+\sum_DX_i) \pmod 2$.

We prove a theorem similar to the case of all honest parties. It should be clear that there is no way for the honest parties to determine whether the dishonest parties act as $n-k$ independent parties each holding one qubit or whether they have colluded to one party. For example, the state $\lvert\Phi_0^{k+1}\rangle=\frac{1}{\sqrt{2}}\Big[\lvert\Phi_0^k\rangle\lvert 0\rangle-\lvert\Phi_1^k\rangle\lvert 1\rangle\Big]$, where the $n-k$ dishonest parties hold a single qubit, passes the Test with probability 1, since the dishonest parties can locally map this state to $\ket{\Phi_0^n}$. Hence, the correct security statement must take into account the fact that the dishonest parties may apply some operator on their space.

\begin{theorem}
Let $|\Psi\rangle$ be the state of all $n$ parties. If $\epsilon=\min_U D(U|\Psi\rangle,|\Phi_0^n\rangle)$, where $U$ is an operator on the space of the dishonest parties, then $\Pr[T(\ket{\Psi}] \leq1-\frac{\epsilon^2}{4}$.
\end{theorem}

\begin{proof}
Let us assume, without loss of generality, that the $n$ parties share a state of the form
\begin{equation}\label{eq:psigeneral}
\lvert\Psi\rangle=|\Phi_0^k\rangle|\Psi_0\rangle+|\Phi_1^k\rangle|\Psi_1\rangle+\ket{\cal{X}},
\end{equation}
where the honest part of state $\ket{\cal{X}}$ is orthogonal to both $\ket{\Phi_0^k}$ and $\ket{\Phi_1^k}$ and the dishonest state is not supposed to be normalized. For the dishonest parties, making the Verifier accept the Test is equivalent to guessing the honest output $Y_H$ before announcing their measurement outcomes $Y_D$. The Helstrom measurement can provide the best guess for $Y_H$ by distinguishing the dishonest states for the honest output $Y_H=0$ and $Y_H=1$.  Let $\rho_{n-k} = tr_k \altketbra{\mathcal{X}}$ be the (unnormalized) reduced density operator when the honest parties are traced out of the state $\ket{\mathcal{X}}$. Recall that $P_k = \altketbra{\Phi_0^k} + \frac{1}{2} I_{S_k}$. Hence, for $X_H=0$ we have :
\begin{eqnarray*}
&&\Pr[\text{guess } Y_H | X_H=0 ] \\
 &&=\frac{1}{2}+\frac{1}{2}\Big\|tr_k\big[P_k\otimes I_{n-k}\ket{\Psi}-\big(I-P_k)\otimes I_{n-k}\ket{\Psi}\big]\Big\| \\
                     &&=\frac{1}{2}+\frac{1}{2}\Big\|\lvert\Psi_0 \rangle\langle\Psi_0| + \frac{1}{2}\rho_{n-k}- \left( \lvert\Psi_1\rangle\langle\Psi_1| + \frac{1}{2}\rho_{n-k}\right) \Big\|\\
& &=\frac{1}{2}+\frac{1}{2}\Big\|\lvert\Psi_0 \rangle\langle\Psi_0|- \lvert\Psi_1\rangle\langle\Psi_1| \Big\|\\ \\
& &= \frac{1}{2}+\frac{1}{2}\sqrt{\big(\big\|\ket{\Psi_0}\big\|^2+\big\|\ket{\Psi_1}\big\|^2\big)^2-4\langle\Psi_0|\Psi_1\rangle^2}
\end{eqnarray*}
where we calculated the trace norm as the sum of the absolute values of the eigenvalues of the matrix. Similarly, we can calculate that:
\begin{eqnarray*}
&&\Pr[\text{guess } Y_H | X_H=1 ]  \\
&&= \frac{1}{2}+\frac{1}{2}\sqrt{\big(\big\|\ket{\Psi_+}\big\|^2+\big\|\ket{\Psi_-}\big\|^2\big)^2-4\langle\Psi_+|\Psi_-\rangle^2}
\end{eqnarray*}
We define $p = \|\ket{\Psi_0}\|^2$ and $q = \|\ket{\Psi_1}\|^2$. Let also $\theta$ be the angle between $\ket{\Psi_0}$ and $\ket{\Psi_1}$ such that $\braket{\Psi_0}{\Psi_1}^2= p q\cos^2{\theta}$. We have
$\langle\Psi_+|\Psi_-\rangle^2=1/4\big( \big\||\Psi_0\rangle\big\|^2-\big\||\Psi_1\rangle\big\|^2  \big)^2 = (p-q)^2/4$.We also know that $\|\ket{\Psi_+}\|^2 + \|\ket{\Psi_-}\|^2 = \|\ket{\Psi_0}\|^2 + \|\ket{\Psi_1}\|^2 =p+q\le 1$. Since the value of $X_H$ is random, we have:
\begin{align}\label{probability}
&\Pr[T(\ket{\Psi})]=\nonumber\\
&=\frac{1}{2}\Big(\Pr[\text{guess } Y_H | X_H=0 ]+\Pr[\text{guess } Y_H | X_H=1 ]    \Big)\nonumber\\
                           &= \frac{1}{2}+\frac{1}{4}\Big(\sqrt{(p+q)^2-4pqcos^2\theta}  +\sqrt{(p+q)^2-(p-q)^2}  \Big)\nonumber\\
                          &=\frac{1}{2}+\frac{1}{4}\Big(\sqrt{(p+q)^2-4pqcos^2\theta}  +2\sqrt{pq}  \Big)\nonumber\\
                          &\leq \frac{1}{2}+\frac{1}{4}\Big(\frac{(p+q)^2-4pqcos^2\theta+1+4pq+1}{2}\Big)\nonumber\\
                          &= \frac{1}{2}+\frac{1}{4}\Big(\frac{(p+q)^2+4pqsin^2\theta+2}{2}\Big)\nonumber\\
                          &= 1-\frac{1}{4}\Big(1-\frac{(p+q)^2+4pqsin^2\theta}{2}\Big)
                          \end{align}
We consider now that the dishonest parties can perform a local operation on their state, in order to maximize their cheating probability. Thus, the distance of the dishonest state from the correct one is : $\epsilon = min_UD((I\otimes U)\ket{\Psi},\ket{\Phi_0^n})=min_U\sqrt{1-F^2((I\otimes U)\ket{\Psi},\ket{\Phi_0^n})}$, where by $F(|\psi\rangle,|\phi\rangle)$ we denote the fidelity between two states $|\psi\rangle$ and $|\phi\rangle$. If the reduced density matrices of the honest players of the perfect and the real state are $\sigma_H$ and $\rho_H$ respectively, it holds that there exists a local operation $R$ on the dishonest state such that:
\[
F((I\otimes R)\ket{\Psi},\ket{\Phi_0^n}) = F(\sigma_H,\rho_H)
\]
By considering this operation $R$, the distance becomes:
\[
\epsilon^2 \le 1-F^2((I\otimes R)\ket{\Psi},\ket{\Phi_0^n}) = 1-F^2(\sigma_H,\rho_H)
\]
We have:
\begin{align*}
\sigma_H&=\frac{1}{2}\big(\ketbra{\Phi_0^k}{\Phi_0^k}+\ketbra{\Phi_1^k}{\Phi_1^k}    \big)\\
\rho_H   &=\\
&p\ketbra{\Phi_0^k}{\Phi_0^k}+q\ketbra{\Phi_1^k}{\Phi_1^k}+\sqrt{pq}\cos{\theta}\big(\ketbra{\Phi_0^k}{\Phi_1^k}+\ketbra{\Phi_1^k}{\Phi_0^k} \big)\\
&+(\text{additional terms that include }\ket{\mathcal{X}})
\end{align*}
\begin{align*}
 F^2(\sigma_H,\rho_H)&=tr^2(\sqrt{\sigma_H^{1/2}\rho_H\sigma_H^{1/2}})\nonumber\\
&=tr^2\Big(\sqrt{\frac{1}{2}\begin{bmatrix}
 p & \sqrt{pq}\cos\theta & 0 & \ldots & 0  \\
 \sqrt{pq}\cos\theta & q  & 0 & \ldots & 0  \\
0 & 0 &  0 & \ldots & 0  \\
\vdots & \vdots & \vdots & \ldots & \vdots
 \end{bmatrix}}\Big)
\end{align*}
Then, we diagonalize the matrix keeping only the upper left non-zero part for simplicity. The determinant of the characteristic function is:
\begin{align*}
det(tI-A)&=det \begin{bmatrix}
t-p & \sqrt{pq}\cos\theta\\
\sqrt{pq}\cos\theta  &t-q
\end{bmatrix}\\
    &=(t-p)(t-q)-pq\cos^2\theta\\
    &=t^2-t(p+q)+pq\sin^2\theta
\end{align*}
Hence, the two eigenvalues of the matrix are:
\[t_{1,2}=\frac{p+q\pm\sqrt{(p+q)^2-4pqsin^2\theta}}{2}\]
We then have:
\begin{align*}
& F^2(\sigma_H,\rho_H)=\frac{1}{4}\Big(\sqrt{p+q+\sqrt{(p+q)^2-4pqsin^2\theta}}\\
                                     &\qquad +\sqrt{p+q-\sqrt{(p+q)^2-4pqsin^2\theta}}   \Big)^2\\
     &\qquad=\frac{1}{4} \Big(2p+2q+2\sqrt{(p+q+\sqrt{(p+q)^2-4pqsin^2\theta})}\\
   &\qquad~~\cdot\sqrt{(p+q-\sqrt{(p+q)^2-4pqsin^2\theta})}\Big)\\
     &\qquad=\frac{1}{2} \big(p+q+\sqrt{(p+q)^2-(p+q)^2+4pqsin^2\theta }\big)\\
     &\qquad=\frac{1}{2} \big(p+q+2sin\theta\sqrt{pq}\big),
\end{align*}
which gives
\[
 \epsilon^2\leq 1 - \frac{p+q}{2} - \sqrt{pq}\sin(\theta)
\]
Recall that from Eq. (\ref{probability}) we have:
\begin{align*}
\Pr[T(\ket{\Psi})]&= 1-\frac{1}{4}\Big(1-\frac{(p+q)^2+4pqsin^2\theta}{2}\Big)\nonumber
\end{align*}
We use the fact that $(p+q)^2 \le (p+q)$ since $p+q \le 1$. Similarly, we use $4pq \sin^2{\theta} \le 2\sqrt{pq} \sin{\theta}$ since $2\sqrt{pq} \sin{\theta} \le 1$. From this, we conclude that:
\begin{align*}
\Pr[T(\ket{\Psi})]&= 1-\frac{1}{4}\Big(1-\frac{(p+q)^2+4pqsin^2\theta}{2}\Big)\nonumber \\
                           &\leq 1-\frac{1}{4}\Big(1-\frac{p+q}{2}-sin\theta\sqrt{pq}\Big) \leq 1-\frac{\epsilon^2}{4}
\end{align*}

\end{proof}

\noindent
{\bf Security for all honest parties}\\

Let $C_\epsilon$ be the event that the state used for the computation, which we denote by $\ket{\Psi}$, is such that $\min_U D(U\ket{\Psi},\ket{\Phi_0^n}) \geq \epsilon$, where $U$ is an operator on the space of the dishonest parties.
We prove the following:
\begin{theorem}
For all $\epsilon >0$,
$\Pr[C_\epsilon] \leq  2^{-S}\frac{4n}{k\epsilon^2}.$
\end{theorem}

\begin{proof}
It is not hard to see that the optimal dishonest strategy that maximizes the probability of $C_\epsilon$ is to create at all repetitions of the protocol a state $\ket{\Psi}$ such that $\min_U D(U\ket{\Psi},\ket{\Phi_0^n})=\epsilon$. In fact, if the source creates a state with smaller distance, then this does not contribute to the probability of event $C_\epsilon$; if the source creates a state with distance larger than $\epsilon$, then this increases the chances of unsuccessful cheating if the state is tested, hence decreasing the probability $C_\epsilon$.

Let $C_\epsilon^l$ be the event that the state is used for the computation at the $l$-th repetition of the protocol and is such that $\min_U D(U\ket{\Psi},\ket{\Phi_0^n}) \geq \epsilon$. This implies that in all previous repetitions the states have been tested successfully. We define the following probabilities:
\begin{description}
\item{$P_1$=} Probability that the parties receive $r=0..0$ from the CRS, which means that they use the state. This probability is equal to $2^{-S}$.
\item{$P_2$=} Probability Party $i$, who becomes the Verifier, is honest. This probability is equal to $k/n$.
\end{description}
Let us also define as VH: the event that the Verifier is Honest and as VD: the event that the verifier is dishonest. Note that for all states, $\Pr[T(\ket{\Psi}) | \text{VD}]=1$. Then,
\begin{eqnarray*}
\Pr[C_\epsilon^l]&=&P_1 (1-P_1)^{l-1}\big(1-P_2+P_2 \cdot Pr[T(\ket{\Psi})|\text{VH}]\big)^{l-1} \\
&\leq&2^{-S}(1-2^{-S})^{l-1}\Big(\frac{n-k}{n}+\frac{k}{n}\big(1-\frac{\epsilon^2}{4}\big)\Big)^{l-1}\\
&=&2^{-S}(1-2^{-S})^{l-1}\Big(1-\frac{k\epsilon^2}{4n}\Big)^{l-1}
\end{eqnarray*}
The infinite integral on $l$ provides an upper bound on the total probability $C_\epsilon$.  We have
\begin{eqnarray*}
\Pr[C_\epsilon] & \leq &  \int_0^{\infty}2^{-S}(1-2^{-S})^l\Big(1-\frac{k\epsilon^2}{4n}\Big)^ldl \\
& \leq &  2^{-S} \int_0^{\infty}\Big(1-\frac{k\epsilon^2}{4n}\Big)^ldl\\
& = & 2^{-S}  \frac{-1}{\log (1-\frac{k\epsilon^2}{4n})}\\
& \leq &2^{-S}  \frac{4n}{k\epsilon^2}
\end{eqnarray*}
\end{proof}

\end{appendix}

\begin{thebibliography}{15}
\expandafter\ifx\csname natexlab\endcsname\relax\def\natexlab#1{#1}\fi
\expandafter\ifx\csname bibnamefont\endcsname\relax
  \def\bibnamefont#1{#1}\fi
\expandafter\ifx\csname bibfnamefont\endcsname\relax
  \def\bibfnamefont#1{#1}\fi
\expandafter\ifx\csname citenamefont\endcsname\relax
  \def\citenamefont#1{#1}\fi
\expandafter\ifx\csname url\endcsname\relax
  \def\url#1{\texttt{#1}}\fi
\expandafter\ifx\csname urlprefix\endcsname\relax\def\urlprefix{URL }\fi
\providecommand{\bibinfo}[2]{#2}
\providecommand{\eprint}[2][]{\url{#2}}

\bibitem[{\citenamefont{Greenberger et~al.}(1989)\citenamefont{Greenberger,
  Horne, and Zeilinger}}]{GHZ}
\bibinfo{author}{\bibfnamefont{D.~M.} \bibnamefont{Greenberger}},
  \bibinfo{author}{\bibfnamefont{M.~A.} \bibnamefont{Horne}}, \bibnamefont{and}
  \bibinfo{author}{\bibfnamefont{A.}~\bibnamefont{Zeilinger}}, in
  \emph{\bibinfo{booktitle}{Bell's Theorem, Quantum Theory, and Conceptions of
  the Universe, M. Kafatos (Ed.), Kluwer, Dordrecht}} (\bibinfo{year}{1989}),
  pp. \bibinfo{pages}{69--72}.

\bibitem[{\citenamefont{Mermin}(1990)}]{mermin:prl90}
\bibinfo{author}{\bibfnamefont{N.~D.} \bibnamefont{Mermin}},
  \bibinfo{journal}{Phys. Rev. Lett.} \textbf{\bibinfo{volume}{65}},
  \bibinfo{pages}{1838} (\bibinfo{year}{1990}).

\bibitem[{\citenamefont{Broadbent et~al.}(2009)\citenamefont{Broadbent,
  Fitzsimons, and Kashefi}}]{broadbent:focs09}
\bibinfo{author}{\bibfnamefont{A.}~\bibnamefont{Broadbent}},
  \bibinfo{author}{\bibfnamefont{J.}~\bibnamefont{Fitzsimons}},
  \bibnamefont{and} \bibinfo{author}{\bibfnamefont{E.}~\bibnamefont{Kashefi}},
  in \emph{\bibinfo{booktitle}{Proceedings of the 50th Annual IEEE Symposium on
  Foundations of Computer Science (FOCS)}} (\bibinfo{year}{2009}), pp.
  \bibinfo{pages}{517--526}.

\bibitem[{\citenamefont{Raussendorf and Briegel}(2001)}]{raussendorf:prl01}
\bibinfo{author}{\bibfnamefont{R.}~\bibnamefont{Raussendorf}} \bibnamefont{and}
  \bibinfo{author}{\bibfnamefont{H.~J.} \bibnamefont{Briegel}},
  \bibinfo{journal}{Phys. Rev. Lett.} \textbf{\bibinfo{volume}{86}},
  \bibinfo{pages}{5188} (\bibinfo{year}{2001}).

\bibitem[{\citenamefont{Zhao et~al.}(2012)\citenamefont{Zhao, Yi and Qi, Bing and Lo, Hoi-Kwong}}]{Lo:pra10}
\bibinfo{author}{\bibfnamefont{Y.}~\bibnamefont{Zhao}},
  \bibinfo{author}{\bibfnamefont{B.} \bibnamefont{Qi}}, \bibnamefont{and}
  \bibinfo{author}{\bibfnamefont{H.-K.} \bibnamefont{Lo}},
  \bibinfo{journal}{Phys. Rev. A} \textbf{\bibinfo{volume}{77}},
  \bibinfo{pages}{052327} (\bibinfo{year}{2008}).


\bibitem[{\citenamefont{McKague et~al.}(2011)\citenamefont{McKague, Yang and Scarani}}]{McKague:arxiv12}
\bibinfo{author}{\bibfnamefont{M.}~\bibnamefont{McKague}},
  \bibinfo{author}{\bibfnamefont{T.~H.}~\bibnamefont{Yang}}, \bibnamefont{and}
  \bibinfo{author}{\bibfnamefont{V.}~\bibnamefont{Scarani}},
  \bibinfo{journal}{e-print: arXiv:1203.2976 [quant-ph]}
  (\bibinfo{year}{2012}).


\bibitem[{\citenamefont{Branciard et~al.}(2012)\citenamefont{Branciard,
  Cavalcanti, Walborn, Scarani, and Wiseman}}]{branciard:arxiv11}
\bibinfo{author}{\bibfnamefont{C.}~\bibnamefont{Branciard}},
  \bibinfo{author}{\bibfnamefont{E.~G.} \bibnamefont{Cavalcanti}},
  \bibinfo{author}{\bibfnamefont{S.~P.} \bibnamefont{Walborn}},
  \bibinfo{author}{\bibfnamefont{V.}~\bibnamefont{Scarani}}, \bibnamefont{and}
  \bibinfo{author}{\bibfnamefont{H.~M.} \bibnamefont{Wiseman}},
  \bibinfo{journal}{Phys. Rev. A} \textbf{\bibinfo{volume}{85}},
  \bibinfo{pages}{010301} (\bibinfo{year}{2012}).

\bibitem[{\citenamefont{Wiseman et~al.}(2007)\citenamefont{Wiseman, Jones, and
  Doherty}}]{wiseman:prl07}
\bibinfo{author}{\bibfnamefont{H.~M.} \bibnamefont{Wiseman}},
  \bibinfo{author}{\bibfnamefont{S.~J.} \bibnamefont{Jones}}, \bibnamefont{and}
  \bibinfo{author}{\bibfnamefont{A.~C.} \bibnamefont{Doherty}},
  \bibinfo{journal}{Phys. Rev. Lett.} \textbf{\bibinfo{volume}{98}},
  \bibinfo{pages}{140402} (\bibinfo{year}{2007}).

\bibitem[{\citenamefont{Jones et~al.}(2007)\citenamefont{Jones, Wiseman, and
  Doherty}}]{jones:pra07}
\bibinfo{author}{\bibfnamefont{S.~J.} \bibnamefont{Jones}},
  \bibinfo{author}{\bibfnamefont{H.~M.} \bibnamefont{Wiseman}},
  \bibnamefont{and} \bibinfo{author}{\bibfnamefont{A.~C.}
  \bibnamefont{Doherty}}, \bibinfo{journal}{Phys. Rev. A}
  \textbf{\bibinfo{volume}{76}}, \bibinfo{pages}{052116}
  (\bibinfo{year}{2007}).

\bibitem[{\citenamefont{Bennet et~al.}(2011)\citenamefont{Bennet, Evans,
  Saunders, Branciard, Cavalcanti, Wiseman, and Pryde}}]{bennet:arxiv11}
\bibinfo{author}{\bibfnamefont{A.~J.} \bibnamefont{Bennet}},
  \bibinfo{author}{\bibfnamefont{D.~A.} \bibnamefont{Evans}},
  \bibinfo{author}{\bibfnamefont{D.~J.} \bibnamefont{Saunders}},
  \bibinfo{author}{\bibfnamefont{C.}~\bibnamefont{Branciard}},
  \bibinfo{author}{\bibfnamefont{E.~G.} \bibnamefont{Cavalcanti}},
  \bibinfo{author}{\bibfnamefont{H.~M.} \bibnamefont{Wiseman}},
  \bibnamefont{and} \bibinfo{author}{\bibfnamefont{G.~J.} \bibnamefont{Pryde}},
  \bibinfo{journal}{e-print: arXiv:1111.0739v1 [quant-ph]}
  (\bibinfo{year}{2011}).

\bibitem[{\citenamefont{Wittmann et~al.}(2011)\citenamefont{Wittmann, Ramelow,
  Steinlechner, Langford, Brunner, Wiseman, Ursin, and
  Zeilinger}}]{wittmann:arxiv11}
\bibinfo{author}{\bibfnamefont{B.}~\bibnamefont{Wittmann}},
  \bibinfo{author}{\bibfnamefont{S.}~\bibnamefont{Ramelow}},
  \bibinfo{author}{\bibfnamefont{F.}~\bibnamefont{Steinlechner}},
  \bibinfo{author}{\bibfnamefont{N.~K.} \bibnamefont{Langford}},
  \bibinfo{author}{\bibfnamefont{N.}~\bibnamefont{Brunner}},
  \bibinfo{author}{\bibfnamefont{H.~M.} \bibnamefont{Wiseman}},
  \bibinfo{author}{\bibfnamefont{R.}~\bibnamefont{Ursin}}, \bibnamefont{and}
  \bibinfo{author}{\bibfnamefont{A.}~\bibnamefont{Zeilinger}},
  \bibinfo{journal}{e-print: arXiv:1111.0760v1 [quant-ph]}
  (\bibinfo{year}{2011}).

\bibitem[{\citenamefont{Smith et~al.}(2012)\citenamefont{Smith, Gillett,
  de~Almeida, Branciard, Fedrizzi, Weinhold, Lita, Calkins, Gerrits, Nam
  et~al.}}]{smith:arxiv11}
\bibinfo{author}{\bibfnamefont{D.~H.} \bibnamefont{Smith}},
  \bibinfo{author}{\bibfnamefont{G.}~\bibnamefont{Gillett}},
  \bibinfo{author}{\bibfnamefont{M.~P.} \bibnamefont{de~Almeida}},
  \bibinfo{author}{\bibfnamefont{C.}~\bibnamefont{Branciard}},
  \bibinfo{author}{\bibfnamefont{A.}~\bibnamefont{Fedrizzi}},
  \bibinfo{author}{\bibfnamefont{T.~J.} \bibnamefont{Weinhold}},
  \bibinfo{author}{\bibfnamefont{A.}~\bibnamefont{Lita}},
  \bibinfo{author}{\bibfnamefont{B.}~\bibnamefont{Calkins}},
  \bibinfo{author}{\bibfnamefont{T.}~\bibnamefont{Gerrits}},
  \bibinfo{author}{\bibfnamefont{S.-W.} \bibnamefont{Nam}},
  \bibnamefont{et~al.}, \bibinfo{journal}{Nature Communications}
  \textbf{\bibinfo{volume}{3}}, \bibinfo{pages}{625} (\bibinfo{year}{2012}).

\bibitem[{\citenamefont{Brassard et~al.}(2003)\citenamefont{Brassard,
  Broadbent, and Tapp}}]{brassard:iwads03}
\bibinfo{author}{\bibfnamefont{G.}~\bibnamefont{Brassard}},
  \bibinfo{author}{\bibfnamefont{A.}~\bibnamefont{Broadbent}},
  \bibnamefont{and} \bibinfo{author}{\bibfnamefont{A.}~\bibnamefont{Tapp}}, in
  \emph{\bibinfo{booktitle}{Proceedings of the 8th International Workshop on
  Algorithms and Data Structures}} (\bibinfo{year}{2003}), vol.
  \bibinfo{volume}{2748}, pp. \bibinfo{pages}{1--11}.

\bibitem[{\citenamefont{Bancal et~al.}(2011)\citenamefont{Bancal, Gisin, Liang,
  and Pironio}}]{bancal:prl11}
\bibinfo{author}{\bibfnamefont{J.-D.} \bibnamefont{Bancal}},
  \bibinfo{author}{\bibfnamefont{N.}~\bibnamefont{Gisin}},
  \bibinfo{author}{\bibfnamefont{Y.-C.} \bibnamefont{Liang}}, \bibnamefont{and}
  \bibinfo{author}{\bibfnamefont{S.}~\bibnamefont{Pironio}},
  \bibinfo{journal}{Phys. Rev. Lett.} \textbf{\bibinfo{volume}{106}},
  \bibinfo{pages}{250404} (\bibinfo{year}{2011}).

\bibitem[{\citenamefont{Cavalcanti et~al.}(2011)\citenamefont{Cavalcanti, He,
  Reid, and Wiseman}}]{cavalcanti:pra11}
\bibinfo{author}{\bibfnamefont{E.~G.} \bibnamefont{Cavalcanti}},
  \bibinfo{author}{\bibfnamefont{Q.~Y.} \bibnamefont{He}},
  \bibinfo{author}{\bibfnamefont{M.~D.} \bibnamefont{Reid}}, \bibnamefont{and}
  \bibinfo{author}{\bibfnamefont{H.~M.} \bibnamefont{Wiseman}},
  \bibinfo{journal}{Phys. Rev. A} \textbf{\bibinfo{volume}{84}},
  \bibinfo{pages}{032115} (\bibinfo{year}{2011}).

\bibitem[{\citenamefont{Lo and Chau}(1998)}]{lochau98}
\bibinfo{author}{\bibfnamefont{H.-K.} \bibnamefont{Lo}} \bibnamefont{and}
  \bibinfo{author}{\bibfnamefont{H.~F.} \bibnamefont{Chau}},
  \bibinfo{journal}{Physica D} \textbf{\bibinfo{volume}{120}},
  \bibinfo{pages}{177} (\bibinfo{year}{1998}).

\bibitem[{\citenamefont{Chaum et~al.}(1988)\citenamefont{Chaum, Cr{\'e}peau,
  and Damg{\aa}rd}}]{chaum:stoc88}
\bibinfo{author}{\bibfnamefont{D.}~\bibnamefont{Chaum}},
  \bibinfo{author}{\bibfnamefont{C.}~\bibnamefont{Cr{\'e}peau}},
  \bibnamefont{and}
  \bibinfo{author}{\bibfnamefont{I.}~\bibnamefont{Damg{\aa}rd}}, in
  \emph{\bibinfo{booktitle}{Proceedings of the 20th Annual ACM Symposium on
  Theory of Computing (STOC 1988), New York}}, pp.
  \bibinfo{pages}{11--19}.

\bibitem[{\citenamefont{Ben-Or et~al.}(1988)\citenamefont{Ben-Or, Michael and Cr{\'e}peau, Claude and Gottesman, Daniel and Hassidim, Avinatan and Smith, Adam}}]{ben-or2006}
  \bibinfo{author}{\bibfnamefont{M.}~\bibnamefont{Ben-Or}},
  \bibinfo{author}{\bibfnamefont{C.}~\bibnamefont{Cr{\'e}peau}},
  \bibinfo{author}{\bibfnamefont{D.}~\bibnamefont{Gottesman}},
  \bibinfo{author}{\bibfnamefont{A.}~\bibnamefont{Hassidim}},
  \bibnamefont{and}
  \bibinfo{author}{\bibfnamefont{A.}~\bibnamefont{Smith}}, in
  \emph{\bibinfo{booktitle}{ Proceedings of the 47th Annual IEEE Symposium on
Foundations of Computer Science (FOCS 2006)}}, pp.
  \bibinfo{pages}{249--260}.
\end{thebibliography}
\end{document}